%% file: main.tex
\documentclass{article}
\pdfoutput=1

\usepackage[utf8]{inputenc}

\def\showauthornotes{1}
\def\showkeys{0}
\def\showdraftbox{0}
\def\confversion{0}

\input{macros}
\input{abbrev}

\def\frs{{\calC^{\mathrm{FRS}}}}
\def\efrs{{\mathrm{Enc}^{\mathrm{FRS}}}}
\newcommand{\agr}{{\mathrm{agr}}}

\title{Improved List Size for Folded Reed-Solomon Codes}
\author{Shashank Srivastava\thanks{{\tt DIMACS, Rutgers University and Institute for Advanced Study}. {\tt shashank.srivastava@rutgers.edu}. Work done in part at TTIC and at DIMACS, and supported in part by the NSF grant CCF-2326685 and by DIMACS.}}
\begin{document}
\maketitle

\draftbox

\begin{abstract}
Folded Reed-Solomon (FRS) codes are variants of Reed-Solomon codes, known for their optimal list decoding radius. We show explicit FRS codes with rate $R$ that can be list decoded up to radius $1-R-\epsilon$ with lists of size $\mathcal{O}(1/ \epsilon^2)$. This improves the best known list size among explicit list decoding capacity achieving codes.

We also show a more general result that for any $k\geq 1$, there are explicit FRS codes with rate $R$ and distance $1-R$ that can be list decoded arbitrarily close to radius $\frac{k}{k+1}(1-R)$ with lists of size $(k-1)^2+1$.

Our results are based on a new and simple combinatorial viewpoint of the intersections between Hamming balls and affine subspaces that recovers previously known parameters. We then use folded Wronskian determinants to carry out an inductive proof that yields sharper bounds.
\end{abstract}

\maketitle

\section{Introduction}
\input{intro}

\input{overview}

\input{prelims}
\section{Intersection of affine subspace and Hamming balls}\label{sec:gen_linear}

In this section, we show that the intersection of a low-dimensional affine subspace and a Hamming ball cannot be too large for any code, giving alphabet-independent bounds on the list size. Let us start with the easiest case where we show that a 1-dimensional affine subspace (essentially, a line) intersects Hamming balls of radius $\frac{2\Delta}{3}$ in at most 2 places.

\begin{lemma}\label{lem:dim1}
	Let $\calC$ be a linear code of distance $\Delta$ and blocklength $n$ over alphabet $\F_q^m$, and let $\calH \sub \calC$ be an affine subspace of dimension 1. Then, for any $g\in (\F_q^m)^n$ and integer $k\geq 1$,
	\[
		\abs{\calH \cap \calL \inparen{g,\frac{k}{k+1}\Delta } } \leq k.
	\]
\end{lemma}

\begin{proof}
	Let $\calH = \{f_0 + \alpha \cdot f_1 \suchthat \alpha \in \F_q\}$ for some $f_0$ and $f_1$ in $\calC$, and let $S\sub [n]$ be the set of coordinates where $f_1$ is non-zero. Clearly, $|S| \geq \Delta \cdot n$.
	
	Let $S_h \sub S$ denote the set of coordinates in $S$ where $g$ and $h\in \calH$ agree. Note that for any two distinct $h_1,h_2 \in \calH$, they differ on every coordinate in $S$. This means that for any distinct $h_1, h_2 \in \calH$, the sets $S_{h_1}$ and $S_{h_2}$ are disjoint.
	
	Now for the sake of contradiction, assume there are $k+1$ codewords \[ h_1,h_2,\cdots ,h_{k+1} \in \calH \cap \calL\inparen{g,\frac{k\Delta}{k+1}}\mper \] 
	Then for at least one of these $h_i$, its $S$-agreement with $g$ must be small so that $|S_{h_i}| \leq \frac{|S|}{k+1}$. For this $h_i$, it therefore also holds that its disagreement with $g$ is at least $\frac{k|S|}{k+1} \geq \frac{k\Delta}{k+1}\cdot n$, which contradicts $h_i \in \calL\inparen{g,\frac{k\Delta}{k+1}}$.
\end{proof}

%
%
Using the above lemma, it is easy to show that $m$-folded Reed-Solomon codes are decodable up to $\frac{2}{3}(1-\frac{m}{m-1}R)$ with lists of size 2. By choosing $m$ to be large enough, this radius can be made arbitrarily close to $\frac{2}{3}(1-R)$, which is the best possible radius for decoding with lists of size 2. This recovers a result of \cite{Tamo24}.
\begin{corollary}
	Let $\frs$ be an $m$-folded Reed-Solomon code of blocklength $N=n/m$ and rate $R$. For any $g\in (\F_q^m)^N$, it holds that
	\[
		\left| \calL\inparen{g,\frac{2}{3}\inparen{1-\frac{m}{m-1}R}} \right| \leq 2
	\]
\end{corollary}

\begin{proof}
	We use \cref{thm:lin_alg_rs} to find a 1-dimensional affine subspace $\calH$ of $\frs$ that contains the list $\calL\inparen{g,\frac{2}{3}\inparen{1-\frac{m}{m-1}R}}$. Then,
	\begin{align*}
		\abs{\calL\inparen{g,\frac{2}{3}\inparen{1-\frac{m}{m-1}R}}} &= \abs{\calL\inparen{g,\frac{2}{3}\inparen{1-\frac{m}{m-1}R}} \cap \calH} \\
		&\leq \abs{\calL\inparen{g,\frac{2}{3}\inparen{1-R}} \cap \calH} \\
		&\leq 2 && \text{[\cref{lem:dim1}]}
	\end{align*}
\end{proof}

Next, we generalize \cref{lem:dim1} to deal with affine subspaces of higher dimensions using induction.
\begin{lemma}\label{lem:dimd}
	Let $\calC$ be a linear code of distance $\Delta$ and blocklength $n$ over alphabet $\F_q^m$, and let $\calH \sub \calC$ be an affine subspace of dimension $d$. Then, for any $g\in (\F_q^m)^n$,
	\[
		\abs{\calH \cap \calL \inparen{g,\frac{k}{k+1} \Delta} } \leq k(k+1)^{d-1}.
	\]
\end{lemma}
\begin{proof}
The proof is by induction on $d$. The base case $d=1$ is proved in \cref{lem:dim1}, and now we prove it for $d\geq 2$ while assuming it is true for $d-1$.

Denote $\calH_g = \calH \cap \calL \inparen{g,\frac{k}{k+1} \Delta}$.

	Let $\calH = \{f_0 + \alpha_1 \cdot f_1 + \cdots +\alpha_d \cdot f_d \suchthat \alpha_i \in \F_q, \forall i \in [d]\}$ for some $f_0, f_1, \cdots ,f_d$ in $\calC$. Let $S\sub [n]$ be the set of coordinates where at least one of $f_1, \cdots, f_d$ is non-zero. By the distance of the code, $|S| \geq \Delta n$. 
	
	As before, we define $S_h \sub S$ to be the set of coordinates in $S$ where $g$ and $h \in \calH$ agree.

	Next, we would like an analog of the disjointness property for agreement sets $\{S_h\}_{h\in \calH}$. We claim that any coordinate $i \in S$ will appear in at most $k(k+1)^{d-2}$ sets in $\{S_h\}_{h\in \calH_g}$. This is because every $h\in \calH_g$ whose $S_h$ contains $i$ must have $h_i = g_i$, and so the collection of these $h$ are restricted to a $(d-1)$-dimensional affine subspace inside $\calH$. By the inductive hypothesis, the number of such $h$ is at most $k(k+1)^{d-2}$. Therefore,
	\[
		\sum_{h \in \calH_g} |S_h| \leq k(k+1)^{d-2} \cdot |S|.
	\]
	It is easy to observe that every $h \in \calH_g$ must have $|S_h| > \frac{|S|}{k+1}$. If not, $g$ and $h$ disagree on at least $ \frac{k}{k+1} |S|$ positions, which is at least $\frac{k}{k+1} \Delta n$, contradicting $h\in \calL \inparen{g,\frac{k}{k+1} \Delta}$. Combining the two,
	\begin{gather*}
		k(k+1)^{d-2} \cdot |S| \geq \sum_{h \in \calH_g} |S_h| > \sum_{h \in \calH_g} \frac{|S|}{k+1} = |\calH_g| \frac{|S|}{k+1} \\
		|\calH_g| < k(k+1)^{d-1}
	\end{gather*}
\end{proof}

\section{Getting more out of the Folded RS code}

The key idea we used in the previous section was that fixing any coordinate to be in the agreement set reduces the search space dimension by 1. However, here we only used agreement of $g$ with a Reed-Solomon codeword, whereas we have the opportunity to decrease the dimension much more by using the agreement of $g$ with a codeword on the \emph{folded} symbol. In an ideal case, such a fixing will uniquely determine the codeword, giving us disjointness of agreement sets as in the case of \cref{lem:dim1} and an optimal list size.

Unfortunately, the set of $m$ constraints imposed by an $m$-folded symbol need not be linearly independent. In fact, since we are considering constraints on a $d$-dimensional space, the best we can hope for are $d$ linearly independent constraints (recall that $m$ is typically chosen so that $m \gg d$). But even $d$ linearly independent constraints need not be guaranteed.

However, these linear dependencies can be bounded in number globally using the Wronskian of (a basis of) the affine subspace we are working with.


Let $\calH$ be an affine subspace of $\F_q[X]^{<Rn}$ with dimension $d$, so that there exist polynomials $h_0, h_1,h_2,\cdots ,h_d$ such that
\[
	\calH = \inbraces{h_0 + \sum_{j=1}^d \alpha_j h_j \suchthat \forall j\in [d], \alpha_j \in \F_q }
\]
Moreover, the set of polynomials $\inbraces{h_1,h_2,\cdots ,h_d}$ is linearly independent over $\F_q$.

The condition that a polynomial $h = h_0 + \sum_{j=1}^d \alpha_j h_j$ agrees with $g$ on position $i\in [N]$ after folding can be written as the collection of $m$ equations: 
\[
	 \forall j\in [m], \quad h(\gamma^{(i-1)m+j-1}) = g(\gamma^{(i-1)m+j-1})
\]
Writing as a linear system,

\begin{align*}
\begin{bmatrix}
h_1(\gamma^{(i-1)m}) & h_2(\gamma^{(i-1)m}) & \cdots & h_d(\gamma^{(i-1)m})\\
h_1(\gamma^{(i-1)m+1}) & h_2(\gamma^{(i-1)m+1}) & \cdots & h_d(\gamma^{(i-1)m+1})\\
\vdots & \vdots & \cdots & \vdots \\
\vdots & \vdots & \cdots & \vdots \\
h_1(\gamma^{(i-1)m+m-1}) & h_2(\gamma^{(i-1)m+m-1}) & \cdots & h_d(\gamma^{(i-1)m+m-1})
\end{bmatrix}
\begin{bmatrix}
\alpha_1 \\
\alpha_2 \\
\vdots \\
\alpha_d
\end{bmatrix}
=
\begin{bmatrix}
(g-h_0)(\gamma^{(i-1)m}) \\
(g-h_0)(\gamma^{(i-1)m+1}) \\
\vdots \\
\vdots \\
(g-h_0)(\gamma^{(i-1)m+m-1})
\end{bmatrix}
\end{align*}

Let us call the $m\times d$ matrix appearing above as $A_i$ for $i\in [N]$, and denote $r_i = \rank(A_i)$. If $r_i$ is always $d$, that is $A_i$ is always full rank, then each agreement $h_i$ and $g_i$ would fix all $\alpha_j$ for $j\in [d]$, and we would get the best case scenario where all agreement sets must be disjoint. However, this need not be true. Guruswami and Kopparty \cite{GK16} used folded Wronskian determinants to show that a weakening of this statement is true in an average sense globally. They wrote this in the language of strong subspace designs, and for completeness we present their proof in our simplified setting.

We first start with the following folded Wronskian criterion for linear independence, whose proof can be found in \cite{GK16}.

\begin{lemma}\label{lem:wronskian}
	Let $\gamma \in \F_q^*$ be a generator. The polynomials $p_1,p_2,\cdots ,p_d \in \F_q[X]^{<Rn}$ are linearly independent over $\F_q$ if and only if the determinant
	\[
		\begin{pmatrix}
			p_1(X) & p_2(X) & \cdots & p_d(X) \\
			p_1(\gamma X) & p_2(\gamma X) & \cdots & p_d(\gamma X) \\
			\vdots & \vdots & \vdots & \vdots \\
			p_1(\gamma^{d-1} X) & p_2(\gamma^{d-1}X) & \cdots & p_d(\gamma^{d-1}X)
		\end{pmatrix}
	\]
	is non-zero as a polynomial in $\F_q[X]$.
\end{lemma}
Next, we use the lemma above to bound the sum of "rank deficit" over all coordinates.
\begin{theorem}[Guruswami-Kopparty \cite{GK16}]
$\sum_{i=1}^N (d-r_i) \leq \frac{d\cdot Rn}{m-d+1}$.
\end{theorem}

\begin{proof}
	We start with instantiating \cref{lem:wronskian} with $p_j = h_j$ for $j\in [d]$, which are linearly independent polynomials used in the definition of $\calH$. By \cref{lem:wronskian}, the determinant of the following matrix
	\[
		H(X) \defeq \begin{bmatrix}
			h_1(X) & h_2(X) & \cdots & h_d(X) \\
			h_1(\gamma X) & h_2(\gamma X) & \cdots & h_d(\gamma X) \\
			\vdots & \vdots & \vdots & \vdots \\
			h_1(\gamma^{d-1} X) & h_2(\gamma^{d-1}X) & \cdots & h_d(\gamma^{d-1}X)
		\end{bmatrix}
	\]
	is non-zero. Denote this determinant by $D(X) = \det(H(X))$. Since each $h_i$ is of degree at most $Rn$, we note that $D(X)$ is a polynomial of degree at most $dRn$, so that the number of zeros of $D(X)$ (with multiplicity) is bounded by $dRn$. Therefore, it suffices to show that the number of zeros of $D(X)$ is at least $(m-d+1)\cdot \sum_{i=1}^N (d-r_i)$.
	
	In fact, we will describe the exact set of zeros with their mutliplicities that illustrates this. The next claim immediately completes the proof. Note that we say that a non-root is a root with multiplicity 0. 
	\begin{claim}\label{claim:root_with_mult}
		For every $i\in [N]$, for every $j\in [m-d+1]$, $\gamma^{(i-1)m+j-1}$ is a root of $D(X)$ with multiplicity at least $d-r_i$.
	\end{claim}
	\end{proof}
	\begin{proof}[Proof of \cref{claim:root_with_mult}]
		Recall that $r_i$ is the rank of matrix $A_i$. For $j\in [m-d+1]$, let $A_{ij}$ denote the $d\times d$ submatrix of $A_i$ formed by selecting all $d$ columns and rows from $j$ to $j+d-1$. That is,
		\begin{align*}
A_{ij} = \begin{bmatrix}
h_1(\gamma^{(i-1)m+j-1}) & h_2(\gamma^{(i-1)m+j-1}) & \cdots & h_d(\gamma^{(i-1)m+j-1})\\
h_1(\gamma^{(i-1)m+j}) & h_2(\gamma^{(i-1)m+j}) & \cdots & h_d(\gamma^{(i-1)m+j})\\
\vdots & \vdots & \cdots & \vdots \\
\vdots & \vdots & \cdots & \vdots \\
h_1(\gamma^{(i-1)m+j+d-2}) & h_2(\gamma^{(i-1)m+j+d-2}) & \cdots & h_d(\gamma^{(i-1)m+j+d-2})
\end{bmatrix}
\end{align*}
	Since $A_{ij}$ is a submatrix of $A_i$, $\rank(A_{ij}) \leq \rank(A_i) = r_i$. If $r_i<d$, then $A_{ij}$ is not full rank and $\det(A_{ij}) = 0$. However, note that $A_{ij} = H(\gamma^{(i-1)m+j-1})$. In conclusion, if $d-r_i>0$, then $\gamma^{(i-1)m+j-1}$ is a root of $D(X)$.
	
	Extending this argument to multiplicities, let $D^{(\ell)}(X)$ be the $\ell^{th}$ derivative of $D(X)$ for $\ell \in \{0,1,\cdots,d\}$. Then this derivative can be written as a sum of $d^{\ell}$ determinants such that every determinant has at least $d-\ell$ columns common with $H(X)$. This follows by writing out the determinant as a signed sum of monomials, applying the product rule of differentiation, and packing them back into determinants. 
	
	Therefore, $D^{(\ell)}(\gamma^{(i-1)m+j-1})$ can be written as a sum of determinants where each determinant has at least $d-\ell$ columns in common with $A_{ij}$. For $\ell = 0,1,\cdots ,d-r_i-1$, this leaves at least $r_i+1$ columns in each determinant from $A_{ij}$. Recall that $\rank(A_{ij})\leq r_i$, which implies that any set of $r_i+1$ columns in $A_{ij}$ are linearly dependent, causing each of the $d^{\ell}$ determinants in the sum for $H^{(\ell)}(\gamma^{(i-1)m+j-1})$ to vanish. We conclude that $H^{(\ell)}(\gamma^{(i-1)m+j-1}) = 0$ for $\ell = 0,1,\cdots , d-r_i-1$, and so $\gamma^{(i-1)m+j-1}$ is a root of $D(X)$ with multiplicity at least $d-r_i$.
	\end{proof}

Now we use the above global upper bound on rank deficit to prove a list size bound with induction.

\begin{theorem}\label{thm:folded_main}
Let $\frs$ be an $m$-folded Reed-Solomon code of blocklength $N=n/m$ and rate $R$. Suppose $d,k,m$ are integers such that $k>d$ and $m \geq k$. 
Then, for any $g \in (\F_q^m)^N$ and for every affine subspace $\calH \sub \frs$ of dimension $d$,
\[
	\abs{ \calH \cap \calL\inparen{g, \frac{k}{k+1} \cdot \inparen{1-\frac{m}{m-k+1} \cdot R}} } \leq (k-1)\cdot d + 1.
\]
\end{theorem}
\begin{proof}
We prove this by induction on $d$. The case $d=0$ is trivial, and the case $d=1$ follows by \cref{lem:dim1} and using $\abs{\calH \cap \calL\inparen{g, \frac{k}{k+1} \cdot \inparen{1-\frac{m}{m-k+1} \cdot R}}} \leq \abs{ \calH \cap \calL\inparen{g, \frac{k}{k+1} \cdot  \inparen{1- R}}}$.

Henceforth, let $d\geq 2$, and denote $\calH_g = \calH \cap \calL\inparen{g, \frac{k}{k+1} \cdot \inparen{1-\frac{m}{m-k+1} \cdot R}}$, and $S_h$ be the agreement set between $g$ and $h$ (over all of $[n]$). Using the lower bound on the size of agreement sets,
\[
	\inparen{\frac{1}{k+1}+\frac{kR}{k+1} \cdot \frac{m}{m-k+1}} N |\calH_g| \leq \sum_{h\in \calH_g} |S_h|
\]
An upper bound on $\sum_{h\in \calH_g} |S_h|$ can again be proved using the inductive hypothesis. Again, we will consider two cases depending on $r_i=0$ or $r_i>0$. In the latter case, we can reduce dimension of the affine space $\calH$ by $r_i>0$ when we decide to assume $h_i= g_i$, so that the inductive hypothesis kicks in. Let $B\sub [N]$ be the bad set with $r_i = 0$, and $b = |B|/N$. It is easy to see that $b < R$.

For $i\in B$, we use the trivial bound $|\calH_g|$ on the number of agreement sets $i$ belongs to. For $i\in \bar{B}$, the dimension reduces to $d-r_i$, and so the coordinate $i$ can appear in at most $(k-1)(d-r_i)+1$ many agreement sets.
\begin{align*}
	\sum_{h\in \calH_g} |S_h| &= \sum_{i=1}^N \abs{ \{ h\in \calH_g \suchthat \forall j\in [t], \  h(\gamma^{(i-1)m+j-1}) = g(\gamma^{(i-1)m+j-1}) \}} \\
	&\leq \sum_{i\in \bar{B}} \insquare{(k-1)(d-r_i) + 1} + \sum_{i\in B}|\calH_g|\\
	&= N-|B|+(k-1)\sum_{i\not\in B} \insquare{d-r_i} + |B|\cdot |\calH_g|\\
	&\leq |B|\cdot |\calH_g| +N-|B| + (k-1)\inparen{\frac{d\cdot Rn}{m-d+1} - d|B|} \\
	&\leq |B|\cdot |\calH_g| + N\inparen{1-b+(k-1)d\inparen{\frac{m}{m-d+1}R - b}}
\end{align*}

Comparing the lower bound and upper bound,
\begin{align*}
	|\calH_g| &\leq \frac{1-b+(k-1)d \inparen{\frac{m}{m-d+1}R-b}}{\inparen{\frac{1}{k+1}+\frac{kR}{k+1} \cdot \frac{m}{m-k+1}-b}} \\
	&< \frac{1-b+(k-1)d \inparen{\frac{m}{m-k+1}R-b}}{\inparen{\frac{1}{k+1}+\frac{kR}{k+1} \cdot \frac{m}{m-k+1}-b}} && [\text{Using }d<k]
\end{align*}
We show that $|\calH_g| < 1+(k-1)d$ by showing that 
\[
\inparen{\frac{1}{k+1}+\frac{kR}{k+1} \cdot \frac{m}{m-k+1}-b} \inparen{|\calH_g| - 1 - (k-1)d}<0
\]
This suffices to conclude our induction.
\begin{align*}
&\inparen{\frac{1}{k+1}+\frac{kR}{k+1} \cdot \frac{m}{m-k+1}-b} \inparen{|\calH_g| - 1 - (k-1)d}\\
&< 1+ \frac{m}{m-k+1} (k-1)dR - \frac{1}{k+1} - \frac{kR}{k+1} \cdot \frac{m}{m-k+1} - \frac{(k-1)d}{k+1} - \frac{kR}{k+1} \cdot \frac{m}{m-k+1}\cdot (k-1)d \\
&= \frac{k}{k+1}\left(1-\frac{m}{m-k+1}R\right) + \frac{m}{m-k+1} (k-1)dR - \frac{(k-1)d}{k+1} - \frac{kR}{k+1} \cdot \frac{m}{m-k+1}\cdot (k-1)d \\
&= \frac{k}{k+1}\left(1-\frac{m}{m-k+1}R\right) - \frac{(k-1)d}{k+1} + \frac{R}{k+1} \cdot \frac{m}{m-k+1}\cdot (k-1)d \\
&=  \frac{k}{k+1}\left(1-\frac{m}{m-k+1}R\right) - \frac{(k-1)d}{k+1}\inparen{1 - \frac{m}{m-k+1}R}\\
&=  \left( \frac{k-(k-1)d}{k+1} \right) \cdot \inparen{1-\frac{m}{m-k+1}R}
\end{align*}
The last term is $\leq 0$ as long as $k\leq (k-1)d$, which is always true for $d\geq 2$.
\end{proof}

This immediately leads to the following corollary.
\begin{corollary}
	Let $\frs$ be an $m$-folded Reed-Solomon code of blocklength $N=n/m$ and rate $R$. For any integer $k$, $1\leq k\leq m$, and for any $g\in (\F_q^m)^N$, it holds that
	\[
		\left| \calL\inparen{g,\frac{k}{k+1}\inparen{1-\frac{m}{m-k+1}R}} \right| \leq (k-1)^2+1
	\]
\end{corollary}

\begin{proof}
	We use \cref{thm:lin_alg_rs} to claim that for $m$-folded RS codes, the list $\calL \inparen{g,\frac{k}{k+1} \cdot \inparen{1-\frac{m}{m-k+1}R}}$ is contained in an affine subspace $\calH \sub \frs$ of dimension $k-1$. Therefore,
	\begin{align*}
		\abs{\calL \inparen{g,\frac{k}{k+1} \cdot \inparen{1-\frac{m}{m-k+1}R}}} &\leq \abs{\calH \cap \calL \inparen{g,\frac{k}{k+1} \cdot \inparen{1-\frac{m}{m-k+1}R}}} \\
		&\leq (k-1)\cdot(k-1)+1 && [\text{\cref{thm:folded_main}}]\\
		&= (k-1)^2+1
	\end{align*}
\end{proof}

\section*{Acknowledgements}
We thank Fernando Granha Jeronimo, Tushant Mittal and Madhur Tulsiani for many helpful discussions. We also thank Prahladh Harsha and Swastik Kopparty, who suggested the problem of deterministic near-linear time algorithm for list decoding of FRS codes during the Simons Institute workshop on Advances in the Theory of Error-Correcting Codes.

\bibliographystyle{alphaurl}
\bibliography{macros,madhur}

\end{document}

%% file: macros.tex
\usepackage{xspace,enumerate}
\usepackage{amsmath,amssymb}
\usepackage{amsthm}
\ifnum\confversion=0
	\usepackage[toc,page]{appendix}
\fi
\usepackage{thmtools}
\usepackage{thm-restate}
\usepackage{color,graphicx}
\usepackage{boxedminipage}
\usepackage{makecell}
\usepackage{tabularx}
\usepackage{relsize}

\ifnum\showkeys=1
\usepackage[color]{showkeys}
\fi

\definecolor{darkred}{rgb}{0.5,0,0}
\definecolor{darkgreen}{rgb}{0,0.35,0}
\definecolor{darkblue}{rgb}{0,0,0.55}

\usepackage[dvipsnames]{xcolor}

\ifnum\confversion=1
\usepackage[pdfstartview=FitH,pdfpagemode=UseNone,colorlinks=false,linkcolor=darkblue,filecolor=darkred,citecolor=darkgreen,urlcolor=darkred,pagebackref,bookmarks=false]{hyperref}
\else
\usepackage[pdfstartview=FitH,pdfpagemode=UseNone,colorlinks,linkcolor=NavyBlue,filecolor=blue,citecolor=OliveGreen,urlcolor=NavyBlue,pagebackref]{hyperref}
\fi

\usepackage[capitalise,nameinlink]{cleveref}
\usepackage[T1]{fontenc}

\usepackage{mathtools,dsfont}

\usepackage{mathpazo}
\usepackage{microtype}
\usepackage[top=1in, bottom=1in, left=1.25in, right=1.25in]{geometry}

\ifnum\showauthornotes=1
\newcommand{\Authornote}[3]{{\sf\small\color{#3}{[#1: #2]}}}
\newcommand{\Authorcomment}[2]{{\sf \small\color{gray}{[#1: #2]}}}
\newcommand{\Authorfnote}[2]{\footnote{\color{red}{#1: #2}}}
\else
\newcommand{\Authornote}[3]{}
\newcommand{\Authorcomment}[2]{}
\newcommand{\Authorfnote}[2]{}
\fi

\ifnum\showdraftbox=1
\newcommand{\draftbox}{\begin{center}
  \fbox{%
    \begin{minipage}{2in}%
      \begin{center}%
        \begin{Large}%
          \textsc{Working Draft}%
        \end{Large}\\
        Please do not distribute%
      \end{center}%
    \end{minipage}%
  }%
\end{center}
\vspace{0.2cm}}
\else
\newcommand{\draftbox}{}
\fi

\newtheorem{theorem}{Theorem}[section]

\newtheorem{definition}[theorem]{Definition}

\newtheorem{lemma}[theorem]{Lemma}

\newtheorem{corollary}[theorem]{Corollary}
\newtheorem{claim}[theorem]{Claim}

\newtheorem{algo}[theorem]{Algorithm}

\def\FullBox{\hbox{\vrule width 6pt height 6pt depth 0pt}}

\def\qed{\ifmmode\qquad\FullBox\else{\unskip\nobreak\hfil
\penalty50\hskip1em\null\nobreak\hfil\FullBox
\parfillskip=0pt\finalhyphendemerits=0\endgraf}\fi}

\def\qedsketch{\ifmmode\Box\else{\unskip\nobreak\hfil
\penalty50\hskip1em\null\nobreak\hfil$\Box$
\parfillskip=0pt\finalhyphendemerits=0\endgraf}\fi}

\def\eps{\varepsilon}
\def\epsilon{\varepsilon}

\def\eps{\epsilon}

\def\phi{\varphi}
\def\cal{\mathcal}

\newcommand{\sub}{\ensuremath{\subseteq}}

\newcommand{\defeq}{:=}

\renewcommand{\bar}{\overline} 

\newcommand{\given}{\;\ifnum\currentgrouptype=16 \middle\fi \vert\;}

\newcommand{\mper}{\,.}

\newcommand{\F}{{\mathbb F}}

\usepackage{nicefrac}

\newcommand{\abs}[1]{\ensuremath{\left\lvert #1 \right\rvert}}

\newcommand{\one}{{\mathbf{1}}}

\newcommand{\Esymb}{\mathbb{E}}

\makeatletter

\def\condPE#1#2{%
	\@ifnextchar\bgroup
	{\ConditionalProbabilityRender{\widetilde{\Esymb}}{#1}{#2}}
	{\ProbabilityRender{\widetilde{\Esymb}}{#1 \given #2}}
}

\def\ConditionalProbabilityRender#1#2#3#4{
	\renderwithdist{#1}{#2}{#3 \given #4}	
}

\def\ProbabilityRender#1#2{
  \@ifnextchar\bgroup%
  {\renderwithdist{#1}{#2}}
   {\singlervrender{#1}{#2}}
}
\def\singlervrender#1#2{%
   \ensuremath{\mathchoice
       {{#1}\left[ #2 \right]}
       {{#1}[ #2 ]}
       {{#1}[ #2 ]}
       {{#1}[ #2 ]}
   }
}
\def\renderwithdist#1#2#3{%
   \@ifnextchar\bgroup
   {\superfancyrender{#1}{#2}{#3}}
   {\ensuremath{\mathchoice
      {\underset{#2}{#1}\left[ #3 \right]}
      {{#1}_{#2}[ #3 ]}
      {{#1}_{#2}[ #3 ]}
      {{#1}_{#2}[ #3 ]}
     }
   }
}
\def\superfancyrender#1#2#3#4#5{
   \ensuremath{\mathchoice
      {\underset{#1}{{#1}}\left#4 #3 \right#5}
      {{#1}_{#2}#4 #3 #5}
      {{#1}_{#2}#4 #3 #5}
      {{#1}_{#2}#4 #3 #5}
   }
}
\makeatother

\newfont{\inhead}{eufm10 scaled\magstep1}

\newcommand{\calC}{{\cal C}}

\newcommand{\calH}{{\cal H}}

\newcommand{\calL}{{\cal L}}

\newcommand{\calO}{{\cal O}}

\newcommand{\poly}{{\mathrm{poly}}}

\newcommand{\suchthat}{{\;\; : \;\;}}

\DeclareMathOperator{\rank}{\operatorname {rk}}

\renewcommand{\bar}[1]{\ensuremath{\overline{#1}}}

\newcommand{\inparen}[1]{\left(#1\right)}             
\newcommand{\inbraces}[1]{\left\{#1\right\}}           
\newcommand{\insquare}[1]{\left[#1\right]}             



%% file: abbrev.tex
\DeclareSymbolFont{extraup}{U}{zavm}{m}{n}
\DeclareMathSymbol{\varheart}{\mathalpha}{extraup}{86}
\DeclareMathSymbol{\vardiamond}{\mathalpha}{extraup}{87}

\def\one{\mathbf 1}

\usepackage{tikz}



%% file: intro.tex
Error-correcting codes are objects designed to withstand corruptions that may appear when communicating through a noisy channel. A code $\calC$ of blocklength $n$ over alphabet $\Sigma$ is simply a subset of $\Sigma^n$, and the elements of $\calC$ are called codewords. The code is said to have distance $\Delta$ if the Hamming distance between any two distinct codewords in $\calC$ is at least $\Delta \cdot n$.

A code with distance $\Delta$ has the potential of correcting $\frac{\Delta n}{2}$ errors, as a corrupted codeword with fewer than $\frac{\Delta}{2}$ fraction of errors may be mapped back to the original codeword it was derived from. Another important parameter of a code is its rate, defined as $R = \frac{\log_{|\Sigma|} |\calC|}{n}$. The Singleton bound says that for any code, $\Delta \leq 1-R$, and Reed-Solomon codes are a well-studied family of codes that achieve this tradeoff.

Therefore, for rate $R$ codes, the best fraction of errors one may correct by uniquely mapping a corrupted codeword to a true codeword is at most $\frac{1-R}{2}$. List decoding is a relaxation of unique decoding, where a corrupted codeword can be mapped to a small \emph{list of codewords}. It is known that there exist codes of rate $R$ that can be list decoded up to radius $1-R-\eps$ for arbitrarily small $\eps$, and the list size needed to do so is $\calO(1/\eps)$. 

We say that a family of codes achieves list decoding capacity, if every code in the family with rate $R$ and blocklength $n$ has the property that all Hamming balls of radius $(1-R-\eps)n$ contain at most $\poly(n)$ codewords. Rephrasing the above, we know that there exist codes that achieve list decoding capacity with list size independent of $n$. Originally, such codes were shown to exist using the probabilistic method, and were not explicit.

%% file: overview.tex
\subsection{Folded Reed-Solomon Codes}
%
The first explicit construction of codes achieving list decoding capacity is due to Guruswami and Rudra \cite{GR08}. These codes are based on \emph{folding} the usual Reed-Solomon code. Folding is a simple operation where for each string, $m$ different alphabet symbols in $[q]$ are treated single symbol of $[q^m]$. Therefore, folding transforms a string of $[q]^n$ into a string of $[q^m]^{n/m}$. 

It is not difficult to see that folding preserves the rate, and the distance cannot decrease. If we fold a code on the Singleton bound such as a Reed-Solomon code, then the distance must also be preserved. The main advantage of folding Reed-Solomon (RS) codes is that the list decoding radius improves to beyond what is currently known for Reed-Solomon codes. This list decoding radius can be made larger than $1-R-\eps$ for any $\eps >0$ by choosing $m = \frac{1}{\eps^2}$. This was shown by Guruswami and Rudra \cite{GR08}, building on the work on Parvaresh and Vardy \cite{PV05}. The proof and algorithm were both simplified by Guruswami in \cite{Gur11} using observations of Vadhan \cite{Vadhan12}.

However, it is important that this folding for RS codes is done along a specific algebraic structure. Let $\F_q$ be a field, and $\gamma \in \F_q^*$ be a primitive element of $\F_q$, so that every non-zero element of $\F_q$ can be written as $\gamma^i$ for some integer $i\geq 0$. The Reed-Solomon code of blocklength $n$ and rate $R$ has codewords in one-to-one correspondence with polynomials of degree $<Rn$ in $\F_q[X]$. The codewords for a (full-length) Reed-Solomon code are based evaluations of these polynomials on $\F_q^*$, as given by the encoding map
\[
	f(X) \rightarrow \insquare{ f(1), f(\gamma), \cdots ,f(\gamma^{n-1}) } \in \F_q^n
\]
where $n=q-1$.
The folded Reed-Solomon code also has its codewords in one-to-one correspondence with the same polynomials, but the encoding map changes to
\[
	f(X) \rightarrow \insquare{ \begin{pmatrix}
	f(1) \\
	f(\gamma)\\
	\vdots \\
	f(\gamma^{m-1})
	\end{pmatrix}, 
	\begin{pmatrix}
	f(\gamma^m) \\
	f(\gamma^{m+1})\\
	\vdots \\
	f(\gamma^{2m-1})
	\end{pmatrix},
	\cdots ,\begin{pmatrix}
	f(\gamma^{n-m}) \\
	f(\gamma^{n-m+1})\\
	\vdots \\
	f(\gamma^{n-1})
	\end{pmatrix} } \in (\F_q^m)^{n/m}
\]

The main result from \cite{Gur11} says that if $m=1/\eps^2$, then for any $g\in (\F_q^m)^{n/m}$, the list $\calL(g,1-R-\eps)$ is contained in an affine subspace of dimension at most $\calO(1/\eps)$. This immediately gives an upper bound of $n^{\calO(1/\eps)}$ for the list size, proving that folded RS codes combinatorially achieve list decoding capacity. \cite{Gur11} also showed that a basis for the affine subspace can be found in $\calO(n^2)$ time, and so this list decoding can be done efficiently.

\subsection{Further Improvements}
Note that the bounds on both list size and alphabet size are large polynomials in $n$. \cite{GR08} showed in their original paper on capacity achieving codes that the Alon-Edmonds-Luby (AEL) distance amplification can also be used for alphabet size reduction to a constant independent of $n$ (but dependent on $\eps$). For list size improvement, \cite{Gur11} isolated a pseudorandom object called \emph{subspace evasive sets} such that no affine subspace of small dimension can intersect with a subspace evasive set in more than $\calO_{\eps}(1)$ points. Thus, if the message polynomials for folded RS codes were chosen from such a set, the lists would be of size at most $\calO_{\eps}(1)$. Moreover, \cite{Gur11} showed the existence of such subspace evasive sets with large enough size so that the loss in rate due to pre-encoding is negligible.

Explicit subspace evasive sets were then constructed by Dvir and Lovett \cite{DL12}, giving codes decodable up to $1-R-\eps$ with list size $(1/\eps)^{\calO(1/\eps)}$. There have also been attempts to use algebraic-geometric (AG) codes \cite{Gur09, GX12, GX22}, variants of subspace evasive-ness \cite{GX13, GK16, GRZ21}, and tensoring \cite{HRW20, KRRSS21} to reduce the alphabet size, list size and/or decoding time.

Somewhat surprisingly, it was shown by Kopparty, Ron-Zewi, Saraf and Wootters \cite{KRSW23} that folded RS codes themselves, without any modification, have much smaller list sizes than previously thought. They proved an upper bound of $(1/\eps)^{\calO(1/\eps)}$ using a general theorem on the intersection of Hamming balls and affine subspaces, matching the list size obtained by \cite{DL12} using subspace evasive sets. Their proof was recently simplified by Tamo \cite{Tamo24}, and was based on earlier ideas on subspace designs from \cite{GK16}.

\section{Our Results}
We extend the above line of work to improve the list size of folded RS codes to $\calO(1/\eps^2)$ for decoding up to radius $1-R-\eps$. To the best of our knowledge, this is the best known list size among explicit capacity achieving codes. This brings the list size of folded RS codes significantly closer to the best possible list size $\Omega(1/\eps)$ when decoding up to $1-R-\eps$.

First, we give an elementary proof that generalizes the results of \cite{KRSW23, Tamo24}. This is again based on upper bounds on the intersection of Hamming balls and affine subspaces, and gives the same asymptotic bound of $(1/\eps)^{\calO(1/\eps)}$ that was known before. 

For a code $\calC$ of blocklength $n$ and alphabet $\Sigma$, and $g\in \Sigma^n$, we use $\calL(g,\eta)$ to denote the list of codewords in $\calC$ at distance less than $\eta$ from $g$.
\begin{theorem}\label{thm:gen_lin}
	Let $\calC$ be a linear code of distance $\Delta$ and blocklength $n$ over alphabet $\F_q$, and let $\calH \sub \calC$ be an affine subspace of dimension $d$. Then, for any $g\in \F_q^n$,
	\[
		\abs{\calH \cap \calL \inparen{g,\frac{k}{k+1} \Delta} } \leq k(k+1)^{d-1}.
	\]
\end{theorem}
For $m$-folded Reed-Solomon codes, it is known that every list of codewords in a ball of radius $\frac{k}{k+1}\cdot \inparen{1-\frac{m}{m-k+1}R}$ is contained in an affine subspace of dimension $k-1$. Thus, such a list will be of size at most $k(k+1)^{k-2}$.

We show that for the specific case of folded RS codes, this analysis can be significantly tightened.

\begin{theorem}\label{thm:folded_rs}
For $m$-folded Reed Solomon codes, and any integer $k \in [m]$,
\[
	\abs{\calL \inparen{g,\frac{k}{k+1} \cdot \inparen{1-\frac{m}{m-k+1}R}}} \leq (k-1)^2 + 1
\]
where $g$ is an arbitrary element of $(\F_q^m)^{n/m}$.
\end{theorem}
By choosing $m \gg k/\eps$, we get that the list size for decoding up to $\frac{k}{k+1}(1-R-\eps)$ is at most $(k-1)^2+1$. For example, if we were constrained to deal with an output list of size at most 50, this theorem shows that we can approach a decoding radius of $\frac{8}{9}(1-R)$ by increasing $m$. We also remark that the decoding radius of $\frac{k}{k+1}(1-R)$ is larger than the Johnson bound $1-\sqrt{R}$ whenever $R\geq \frac{1}{k^2}$.

\subsection{Intersections of Hamming balls and Affine Subspaces}
%
As mentioned before, \cref{thm:gen_lin} implies that when decoding $m$-folded RS codes up to $ \approx \frac{k}{k+1}\inparen{1-R}$, where $m$ is sufficiently large compared to $k$, the list size is bounded by $k\cdot (k+1)^{k-2}$. 

This implies the results of \cite{KRSW23, Tamo24} for capacity achieving codes, but also works for fixed small values of $k=2,3,\cdots$. The case $k=1$ is just unique decoding. We note that the case of $k=2$ and the corresponding list size of 2 was also shown by \cite{Tamo24}, but his method did not generalize to $k>3$. 

We use many of the same techniques as earlier works, but structure our proof in a bottom-up inductive approach instead of a top-down random pinning argument . We start by showing a simple combinatorial argument that shows that an affine subspace of dimension 1, or a line, can intersect a Hamming ball of radius $\frac{s}{s+1}\Delta$ in at most $s$ points. This relies on the simple observation that given a line, $[n]$ can be divided into two sets $S$ and $\overline{S}$ such that all the points on the line agree on $\overline{S}$, and any two points on the line differ \emph{everywhere} on $S$. Thus, the restriction to $S$ can be seen as a distance 1 code, and moreover $|S| \geq \Delta n$.

Therefore, the agreement sets between codewords and the received word must be disjoint over $S$, and if there were $s+1$ codewords in the list, one of these agreement sets must have size at most $\frac{|S|}{s+1}$ (when restricted to $S$). This codeword and the received word differ in at least $\frac{s}{s+1}\cdot |S| \geq \frac{s}{s+1}\Delta n$ positions, contradicting its membership in the list.

We then use an induction on the dimension of the affine subspace, and we can conclude an upper bound on list size of $(1/\eps)^{1/\eps}$ for decoding up to $1-R-\eps$ by choosing $k\approx 1/\eps$ and $m\approx 1/\eps^2$.

\subsection{Improvements using the Folded Structure}
Until now, our results are based on the above general argument applied to the RS code underlying the folded RS code. Rephrasing the argument for 1-dimensional case above, we can say that fixing a single coordinate in $S$ to be an agreement determines the entire codeword. Therefore, each coordinate may appear in only one agreement set. For higher dimensions, fixing a coordinate to be an agreement fixes its value, and therefore reduces the dimension of the affine subspace by 1. We then use the inductive hypothesis to obtain a bound on the number of agreement sets any coordinate belongs to. This acts as a weak version of the disjointness property, and a double counting argument similar to 1-dimensional case finishes the proof.

%

For a folded RS code, fixing a particular coordinate to be an agreement actually gives us multiple equations, and can reduce the dimension by much more than 1. In an ideal case, all of these equations will be linearly independent, and we can fix the entire codeword after fixing a single coordinate. This would allow us to conclude that the agreement sets are disjoint, as in the 1-dimensional case above. Unfortunately, such a strong guarantee need not hold. 

However, \cite{GK16} proved a global upper bound on the sum of rank deficit we see in each coordinate. This is based on the notion of a folded Wronskian determinant criterion for linear independence, and this part of the proof was also used by \cite{KRSW23, Tamo24}. However, with our sharper bottom-up structure of the proof, we are able to improve the list size to $(k-1)^2+1$. 

This gives the optimal list size bound of 2 for decoding up to $\frac{2}{3}(1-R)$, and in the regime of list decoding capacity with $k\approx 1/\eps$, the list size is bounded by $\calO(1/\eps^2)$.

\subsection{Discussion}
%
While our results are written for folded RS codes, the exact same machinery also applies for univariate machinery codes - the list is contained in an affine subspace of constant dimension, and the determinant of the corresponding Wronskian matrix is non-zero as a polynomial. This allows us to get a similar bound on list size of univariate multiplicity codes. To simplify presentation, we choose to focus only on folded RS codes in this paper.

One advantage of the arguments of \cite{KRSW23, Tamo24} is that they immediately suggest randomized algorithms to find the list in linear time, given a basis for the affine subspace. 
One wonders whether our proof technique can be used to give a deterministic near-linear time algorithm to obtain the list given a basis for the affine subspace in which it is contained. If true, this would give a near-linear time \emph{deterministic} algorithm for decoding folded RS codes using the work of Goyal, Harsha, Kumar, and Shankar \cite{GHKS24}.

Indeed, when decoding up to $\frac{2}{3}(1-R)$, which means we are dealing with a 1-dimensional affine subspace, a simple near-linear time deterministic algorithm can be obtained. If the affine subspace is $\{f_0 + \alpha f_1 \suchthat \alpha \in \F_q\}$, and the received word is $g$, we use the two most frequent values appearing among $\{\frac{g(i)-f_0(i)}{f_1(i)} \}$ over $i$ such that $f_1(i)\neq 0$. This avoids having to try all possible values in $\F_q$ for $\alpha$, which naively would require quadratic time. Can this idea be generalized to higher dimensional affine subspaces?

Finally, the notion of Wronskian determinants is tailored to the algebraic structure of folded RS and multiplicity codes. Can we generalize it to general linear codes, and what further applications could it have?

\subsection{Concurrent Work}

The results in this work originally appeared in the author's thesis \cite{Sri24}, where it was posed as an open problem to improve the list size to optimal $(k-1)+1=k$ instead of $(k-1)^2+1$.
Around the same time, in an independent work, Chen and Zhang \cite{CZ24} showed that (explicit) folded Reed-Solomon codes indeed have this optimal list size property, exhibiting the first proof of such a result among explicit codes. Their proof also uses the properties of Wronskian matrices established in \cite{GK16}, but otherwise uses different tools compared to this work to get the optimal list size.

%% file: prelims.tex
\section{Preliminaries}

\begin{definition}[Distance and agreement]
	Let $\Sigma$ be a finite alphabet and let $f,g\in \Sigma^n$. Then the (fractional) distance
        between $f,g$ is defined as \[ \Delta(f,g) = \mathbb{E}_{i\in [n]}\insquare{ \one \inbraces{f_i \neq g_i} }\mper \]
    Likewise, the (fractional) agreement between $f,g$ is defined as \[ \agr(f,g) = \mathbb{E}_{i\in [n]}\insquare{ \one \inbraces{f_i = g_i} }\mper \]
\end{definition}

Throughout this paper, we will always use distance and agreement to mean fractional distance and fractional agreement respectively.
\begin{definition}[Code, distance and rate]
	A code $\calC$ of block length $n$, distance $\delta$ and rate $R$ over an alphabet size $\Sigma$ is a set $\calC \subseteq \Sigma^n$ with the following properties
	\begin{enumerate}[(i)]
		\item $R = \frac{\log_{|\Sigma|} |\calC|}{n}$
		\item $\delta = \min_{\substack{h_1,h_2\in \calC \\ h_1\neq h_2}} \Delta(h_1,h_2)$
	\end{enumerate}
\end{definition}

\begin{definition}[List of codewords]
	Let $\calC$ be a code with alphabet $\Sigma$ and blocklength $n$. Given $g\in \Sigma^n$, we use $\calL(g,\eta)$ to denote the list of codewords from $\calC$ whose distance from $g$ is less than $\eta$. That is,
	\[
		\calL(g,\eta) = \inbraces{ h\in \calC \suchthat \Delta(g,h) <\eta}\mper
	\]
\end{definition}

	We say that a code is combinatorially list decodable up to radius $\eta$ if for every $g\in \Sigma^n$, $\calL(g,\eta)$ is of size at most $\poly(n)$. Likewise, we say a code is efficiently list decodable up to radius $\eta$ if it is combinatorially list decodable up to $\eta$, and the list $\calL(g,\eta)$ can be found in time $\poly(n)$.
\begin{definition}[Folded Reed-Solomon Codes]
	Let $\F_q$ be a field with $q>n$, and $\gamma$ be an element of order at least $n$. The encoding function for the $m$-folded Reed-Solomon code $\frs$ with rate $R$, blocklength $N = n/m$, alphabet $\F_q^m$, and distance $1-R$ is denoted as $\efrs : \F_q[X] \rightarrow (\F_q^m)^N$, given by
	\[
		f(X) \rightarrow \insquare{ \begin{pmatrix}
	f(1) \\
	f(\gamma)\\
	\vdots \\
	f(\gamma^{m-1})
	\end{pmatrix}, 
	\begin{pmatrix}
	f(\gamma^m) \\
	f(\gamma^{m+1})\\
	\vdots \\
	f(\gamma^{2m-1})
	\end{pmatrix},
	\cdots ,\begin{pmatrix}
	f(\gamma^{n-m}) \\
	f(\gamma^{n-m+1})\\
	\vdots \\
	f(\gamma^{n-1})
	\end{pmatrix} } \in (\F_q^m)^{N}
	\]
	The code $\frs$ is given by
	\[
		\frs = \inbraces{ \efrs(f(X)) \suchthat f(X) \in \F_q[X]^{<Rn}}
	\]
\end{definition}

We will also use $\efrs(S)$ to denote the set of all $\efrs(s)$ for $s\in S$. Under this notation, $\frs = \efrs(\F_q^{<Rn})$.

An $m$-folded RS code with $m=1$ is called the (full-length) Reed-Solomon code.
We will be needing the following main result of \cite{Gur11} throughout the paper.

\begin{theorem}[\cite{Gur11}]\label{thm:lin_alg_rs}
	Let $\frs$ be an $m$-folded Reed-Solomon code of blocklength $N=n/m$ and rate $R$. For any integer $k$, $1\leq k\leq m$, and for any $g\in (\F_q^m)^N$, there exists a affine subspace $\calH$ of $\F_q[X]^{<Rn}$ of dimension $k-1$ such that
	\[
		\calL\inparen{g,\frac{k}{k+1}\inparen{1-\frac{m}{m-k+1}R}} \sub \efrs(\calH)
	\]
	In particular, 
	\[
		\left| \calL\inparen{g,\frac{k}{k+1}\inparen{1-\frac{m}{m-k+1}R}} \right| \leq \left| \efrs(\calH) \right| = |\calH| = q^{k-1}
	\]
	Moreover, a basis for $\calH$ can be found in time $\calO((Nm\log q)^2)$.
\end{theorem}
	We mention that if $\calH$ is an affine subspace of $\F_q[X]^{<Rn}$ with dimension $k-1$, then $\efrs(\calH)$ is an $\F_q$-linear subspace of $\frs$ of dimension $k-1$. We will therefore use list-containing affine subspaces in both $\F_q[X]^{<Rn}$ and $\frs$.
